\documentclass[journal]{IEEEtran}
\IEEEoverridecommandlockouts

\usepackage[english]{babel}
\usepackage{amsmath,amsthm,amssymb,bbm}
\usepackage{cite,cleveref,enumitem}
\usepackage{algorithmicx,algorithm,algpseudocode}
\usepackage{color}

\newtheorem{theorem}{Theorem}

\newtheorem{lemma}[theorem]{Lemma}

\newcommand{\bs}[1]{\boldsymbol{#1}}
\newcommand{\cl}[1]{\mathcal{#1}}
\newcommand{\bb}[1]{\mathbb{#1}}
\newcommand{\bbm}[1]{\mathbbm{#1}}

\newcommand{\lb}{\left(}
\newcommand{\rb}{\right)}
\newcommand{\ls}{\left[}
\newcommand{\rs}{\right]}
\newcommand{\lc}{\left\{}
\newcommand{\rc}{\right\}}

\newcommand{\lv}{\left\vert}
\newcommand{\rv}{\right\vert}
\newcommand{\lV}{\left\Vert}
\newcommand{\rV}{\right\Vert}

\newcommand{\diag}{\mathrm{diag}}
\newcommand{\T}{\mathsf{T}}
\crefname{figure}{Fig.}{Figs.}  
\crefname{section}{Sec.}{Secs.}

\begin{document}

\title{Convergence of Expectation-Maximization Algorithm with Mixed-Integer  Optimization}

\author{Geethu Joseph
\thanks{G. Joseph is with the Faculty of Electrical Engineering, Mathematics, and Computer Science, Delft University of Technology, Delft 2628 CD, The Netherlands.  
(email: g.joseph@tudelft.nl)}}

\markboth{Journal of \LaTeX\ Class Files}
{G. Joseph \MakeLowercase{\textit{et al.}}: Convergence Analysis of a Special Class of Expectation-Maximization Algorithm with Mixed-Integer Maximum Likelihood Optimization}
\maketitle
\begin{abstract}
The convergence of expectation-maximization (EM)-based algorithms typically requires continuity of the likelihood function with respect to all the unknown parameters (optimization variables). The requirement is not met when parameters comprise both discrete and continuous variables, making the convergence analysis nontrivial. This paper introduces a set of conditions that ensure the convergence of a specific class of EM algorithms that estimate a mixture of discrete and continuous parameters. Our results offer a new analysis technique for iterative algorithms that solve mixed-integer non-linear optimization problems. { As a concrete example, we prove the convergence of an existing EM-based sparse Bayesian learning algorithm that estimates the state of a linear dynamical system with jointly sparse inputs and bursty missing observations.} Our results establish that the algorithm converges to the set of stationary points of the maximum likelihood cost with respect to the continuous optimization variables. 
\end{abstract}

\section{Introduction}\label{sec:intro}
The Expectation-Maximization (EM) algorithm is a general technique for maximum likelihood or maximum a posteriori estimation~\cite{ghahramani1993function,ghahramani1998learning}. It is a crucial ingredient in well-known algorithms like Baum-Welch~\cite{baum1970maximization}, inside-outside~\cite{lari1990estimation}, sparse Bayesian learning (SBL)~\cite{wipf2004sparse}, and their numerous variants. EM's popularity is due to its simplicity, stability (monotonic increase in likelihood),  and convergence guarantees for many statistical problems.  The convergence analysis of EM, presented in~\cite{wu1983convergence}, establishes conditions under which EM converges to a stationary point of the likelihood function. The literature also offers convergence analyses for specific cases, such as EM for Gaussian mixtures~\cite{xu1996convergence} and EM with squared iterative methods~\cite{varadhan2008simple}. However, these analyses assume that the likelihood function is continuous in all unknown hyperparameters, implying that the parameters belong to an open set and are all continuous. This assumption may not hold in general. For example, the EM-based SBL framework~\cite{joseph2022state} is used to estimate sparse state of a linear dynamical with missing observations, where the unknown state vector belongs to an open set, while the unknown missing data status is discrete (either missing or not missing). { Some motivating applications for this problem include identifying missing data indices in occlusions due to nonlinear energy harvesting or environmental factors in 
wireless networks (motion tracking~\cite{song2011feedback}, network traffic reconstruction~\cite{roughan2011spatio}, localization refinement~\cite{rallapalli2010exploiting}, urban traffic sensing improvement~\cite{li2011compressive}, and structural health monitoring~\cite{ji2014method,thadikemalla2017simple}), satellite imaging systems~\cite{carlavan2012satellite,scalise2008measurement}, and downlink channel estimation via feedback through a bursty channel~\cite{joseph2022state}.} Despite demonstrating good empirical performance of the EM algorithm~\cite{joseph2022state}, its convergence becomes nontrivial due to discrete parameters, posing a challenge to existing EM convergence results. Insipred by this setting, we study the convergence of the EM algorithm with both continuous and discrete hyperparameters.

The contributions of this paper are twofold. First, we relax the assumption that the set of unknowns estimated by EM is purely continuous, allowing for a general set that comprises both continuous and discrete parameters. We derive mild conditions ensuring the convergence of EM to a stationary point of the likelihood function. Notably, when the unknowns belong to an open set, our results reduce to those of~\cite{wu1983convergence}. Second, we apply these results to establish the convergence of the EM-based SBL algorithm presented in~\cite{joseph2022state}.

\section{Statistical Model and Convergence Result}
Consider the statistical model which generates observations $\bs{Y}$, unobserved latent data $\bs{X}$, and unknown parameters $\bs{\theta}^*\in\bs{\Theta}$. We assume that a part of the parameter is continuous and the other part is discrete and finite, i.e., $\bs{\Theta}\subseteq\bb{G}\times \bb{H}$ and
\begin{equation}\label{eq:theunknown}
\bs{\theta}^*=\begin{bmatrix}
\bs{\gamma}^{\T}&{ \bs{\alpha}}^{\T}
\end{bmatrix}^{\T},
\end{equation} where $\bs{\gamma}\in\bb{G}$ and ${ \bs{\alpha}}\in \bb{H}$. Here, $\bb{G}\subseteq\bb{R}^N$ is a set without isolated points, and $\bb{H}\subseteq\bb{R}^K$ is a finite (countable) set. For the ease of exposition, we use $\bs{\theta}$ and $(\bs{\gamma},{ \bs{\alpha}})$ interchangeably.

Let $p(\bs{Y},\bs{X};\bs{\theta}^*)$ be the joint distribution of data $(\bs{Y},\bs{X})$ conditioned on the parameter $\bs{\theta}^*$. Then, the maximum likelihood (ML) estimate of the unknown $\bs{\theta}^*$ is
\begin{equation}\label{eq:overall_opt}
\underset{\bs{\theta}\in\bs{\Theta}}{\arg\max}\;L(\bs{\theta}),
\end{equation}
where $L(\bs{\theta})$ is the likelihood function given by 
\begin{equation}
    L(\bs{\theta})=\log p(\bs{Y};\bs{\theta}) = \bb{E}_{\bs{X} }\{\log p(\bs{Y},\bs{X};\bs{\theta})\}.
\end{equation}
If the above optimization is not tractable, we can use the EM algorithm to solve \eqref{eq:overall_opt}. The EM algorithm is an iterative algorithm, with each iteration comprising an expectation step (E-step) and a maximization step (M-step). Let $\bs{\theta}^{(r)}=[\bs{\gamma}^{(r)\T}\;\;\;{ \bs{\alpha}}^{(r)\T}]^{\T}$ be the EM iterate in the $(r-1)$th iteration. In the $r$th iteration, the E-step computes the expected log-likelihood of $\bs{\theta}$ with respect to the distribution of $\bs{X}$ conditioned on the current iterate $\bs{\theta}^{(r)}$ and observations $\bs{Y}$,
{\begin{equation}\label{eq:Q_defn}
    Q(\bs{\theta};\bs{\theta}^{(r)})= \bb{E}_{\bs{X}|\bs{Y},\bs{\theta}^{(r)} }\{\log p(\bs{Y},\bs{X};\bs{\theta})\}.
\end{equation}
The M-step maximizes this function with respect to $\bs{\theta}$ to obtain the new iterate $\bs{\theta}^{(r+1)}$, i.e.,
\begin{equation}
\underset{\bs{\theta}\in\bs{\Theta}}{\arg\max} \; Q( \bs{\theta};\bs{\theta}^{(r)}).
\end{equation}}
So, the $r$-th EM iteration can be summarized as a mapping $\cl{G}:\bs{\Theta}\to\bs{\Theta}$, i.e.,
\begin{equation}\label{eq:MStep}
\bs{\theta}^{(r+1)} = \cl{G}(\bs{\gamma}^{(r)},{ \bs{\alpha}}^{(r)}) = \underset{\bs{\theta}\in\bs{\Theta}}{\arg\max} \; Q(\bs{\theta};\bs{\theta}^{(r)}).
\end{equation}

Next, we present the main result of the section, which provides a list of sufficient conditions for the EM algorithm to converge to a stationary point of the ML cost function in~\eqref{eq:overall_opt}.
\begin{theorem}\label{thm:general_convergence}
Let $\{\bs{\theta}^{(r)}\}_{r=0}^{\infty}$ be the sequence generated by the EM algorithm, as summarized in \eqref{eq:MStep}, to solve the ML optimization problem in \eqref{eq:overall_opt}. Assume the following conditions,
\begin{enumerate}
\item There exists a constant $C\in\bb{R}$ such that $L(\bs{\theta}^{(0)})\leq C$, for any $\bs{\theta}^{(0)}\in\bs{\Theta}$. \label{assmp:lowerbnd}
\item The level set $\{\bs{\theta}: L(\bs{\theta})\geq L(\bs{\theta}^{(r)})\}$ is compact, for any integer $r > 0$.\label{assmp:compact}
\item For any ${ \bs{\alpha}}_*\in\bb{H}$, the iteration mapping $\cl{G}(
\bs{\gamma}, { \bs{\alpha}}_* )$ in \eqref{eq:MStep} is closed at all values of $[\bs{\gamma}^{\T}\;\;\;  { \bs{\alpha}}_*^{\T}]^{\T}\in\bs{\Theta}$.\label{assmp:closed}
\item The function $Q(\bs{\theta};\bs{\theta}^{(r)})$ is a continuous
function of $\bs{\gamma}$ and $\bs{\gamma}^{(r)}$, for a fixed value of ${ \bs{\alpha}}$ and ${ \bs{\alpha}}^{(r)}$, where $\bs{\theta}$ and $\bs{\theta}^{(r)}$ take the form \eqref{eq:theunknown} and \eqref{eq:MStep}, respectively. \label{assmp:contin}
\end{enumerate}
Then, the sequence of iterates $\{\bs{\theta}^{(r)}\}_{r=1}^{\infty}$ converges to a subset of $\cl{S}_*$ over which $L(\bs{\theta})$ is a constant. Here,
\begin{equation}
\cl{S}_* = \{ \bs{\theta}\in\bs{\Theta}: \bs{\theta} = \begin{bmatrix}
\bs{\gamma}\in\bb{R}^{N} & { \bs{\alpha}}\in\bb{H}
\end{bmatrix}\; \text{and}\; \nabla_{\bs{\gamma}} L(\bs{\theta})=\bs{0} \},
\end{equation}
where $\nabla_{\bs{\gamma}}$ denote the gradient with respect to $\bs{\gamma}$.
\end{theorem}
\begin{proof}
See \Cref{app:general_convergence}.
\end{proof}
Here, Conditions 1 and 2  refer to the likelihood, and Conditions 3 and 4 are linked to the iterate update procedure. Conditions 1, 2, and 4 are similar to (8), (6), and (7) in \cite{wu1983convergence}, respectively. Condition 3 is required for the global convergence theorem~\cite{zangwill1969nonlinear} to hold. Further, our analysis generalizes the EM convergence result in \cite{wu1983convergence}.  Using a similar proof technique, we can extend the other results in \cite{wu1983convergence} to our setting (for example, the results on generalized EM). Furthermore, if the stationary points of the likelihood cost are isolated, the algorithm converges to a single point. Also, we assume that $\bb{H}$ is finite, such as any bounded subset of integers or rational numbers. There are estimation problems where this assumption holds, and we give an example in the next section.

\section{Analysis of Kalman SBL for State Estimation of a Linear System with Missing Outputs}
In this section, we prove convergence guarantees to the EM-based SBL algorithm in ~\cite{joseph2022state} using \Cref{thm:general_convergence}. The algorithm aims to estimate the states of a linear dynamical system with jointly sparse inputs and missing observations at unknown time instants~\cite{joseph2022state}. Specifically, we consider a discrete-time linear dynamical system given by
\begin{equation}\label{eq:lds}
\bs{x}_k = \bs{D}\bs{x}_{k-1}+\bs{u}_k\;\;\;\text{and}\;\;\;
\bs{y}_k = { \bs{\alpha}}_k^* \bs{A}\bs{x}_k+\bs{w}_k.
\end{equation}
Here, $\bs{x}_k\in\bb{R}^n,\bs{u}_k\in\bb{R}^n$, and $\bs{y}_k\in\bb{R}^m$ are the state, input, and observation at time $k$, respectively. Also, $\bs{w}_k\in\bb{R}^m$ is the zero-mean Gaussian distributed measurement noise at time $k$ whose variance is $\sigma^2$. Also, $\bs{D}\in\bb{R}^{n\times n}$ is the state transition matrix, and  $\bs{A}\in\bb{R}^{m\times n}$ is the output matrix. The inputs are jointly sparse, and the initial state $\bs{x}_0=\bs{0}\in\bb{R}^n$. Further, ${ \bs{\alpha}}_k^*\in\{0,1\}$ represents whether the signal part $ \bs{A}\bs{x}_k$ is missing or not in the observation $\bs{y}_k$. The missing data indicator ${ \bs{\alpha}}_k^*$ for $k=1,2,\ldots$ follows a hidden Markov model: for any integer $k>0$ and $i,j\in\{ 0,1\}$, we have $\bb{P}\{{ \bs{\alpha}}_k^*=i|{ \bs{\alpha}}_{k-1}^*=j\} =
p_j$, if $i=j$. For a given integer value of $K<\infty$, the algorithm aims to estimate the state matrix $\bs{X}$ using the output matrix $\bs{Y}$  when ${ \bs{\alpha}}^*$ is unknown, where we define
\begin{align}
\bs{X} &\triangleq \begin{bmatrix}
\bs{x}_1 & \bs{x}_2 & \ldots & \bs{x}_K 
\end{bmatrix}\in\bb{R}^{N\times K}\\
\bs{Y} &\triangleq \begin{bmatrix}
\bs{y}_1 & \bs{y}_2 & \ldots & \bs{y}_K 
\end{bmatrix}\in\bb{R}^{m\times K}\\
{ \bs{\alpha}}^* &\triangleq \begin{bmatrix}
{ \bs{\alpha}}_1^* & { \bs{\alpha}}_2^* & \ldots & { \bs{\alpha}}_K^* 
\end{bmatrix}^{\intercal}\in\{0,1\}^K.
\end{align}

The SBL framework assumes a fictitious Gaussian prior on the sparse vectors with a common matrix with $\bs{\gamma}^*\in\bb{R}_+^{N}$ along the diagonal, i.e., $p(\bs{x}_k|\bs{\gamma}^*)=\cl{N}(\bs{0},\diag\{\bs{\gamma}^*\})$. Then, we jointly estimate $\bs{\theta}^*=[\bs{\gamma}^*\;\;{ \bs{\alpha}}^*]\in\bb{R}_+^{N}\times\{0,1\}^K$. The resulting $Q$ function in the M-step is
\begin{align}\label{eq:Q_fnt_sbl}
Q( \bs{\theta};\bs{\theta}^{(r)})
&= \bb{E}_{ \bs{X}  |\bs{Y},\bs{\theta}^{(r)}}
\{ \log \ls 
p\lb\bs{Y}|\bs{X},{ \bs{\alpha}}\rb p\lb{ \bs{\alpha}}\rb p\lb\bs{X},\bs{\gamma}\rb
\rs\}\\
&=\bb{E}_{\bs{X}|\bs{Y},\bs{\theta}^{(r)}\!\!}\{ \log p\lb\bs{Y}\middle|\bs{X},{ \bs{\alpha}}\rb\}
+\log p\lb{ \bs{\alpha}}\rb\notag\\
&\;\;\;+\bb{E}_{\bs{X}|\bs{Y},\bs{\theta}^{(r)}}\{\log p\lb\bs{X},\bs{\gamma}\rb\}.
\end{align}
So, the optimization problem in the M-step is separable in the $\bs{\gamma}$ and ${ \bs{\alpha}}$. Since the optimization is separable, the mapping $\cl{G}$ in \eqref{eq:MStep} can be decomposed as follows:
\begin{equation}\label{eq:mapping_sbl}
 \begin{bmatrix}
\bs{\gamma}^{(r+1)}\\
{ \bs{\alpha}}^{(r+1)}
\end{bmatrix} \in \cl{G}(
\bs{\gamma}^{(r)}, { \bs{\alpha}}^{(r)}) = \begin{bmatrix}
\cl{G}_{\gamma}\lb \bs{\gamma}^{(r)},{ \bs{\alpha}}^{(r)}\rb \subset\bb{R}^n_+\\
\cl{G}_{{\alpha}}\lb \bs{\gamma}^{(r)},{ \bs{\alpha}}^{(r)}\rb \subset\bb{H}
\end{bmatrix}.
\end{equation}
The resulting algorithm uses the Kalman smoothing to compute $\cl{G}_{\gamma}\lb \bs{\gamma}^{(r)},{ \bs{\alpha}}^{(r)}\rb$ and Viterbi algorithm to $\cl{G}_{{\alpha}}\lb \bs{\gamma}^{(r)},{ \bs{\alpha}}^{(r)}\rb$. 

Further, the set $\cl{G}_{\gamma}\lb \bs{\gamma}^{(r)},{ \bs{\alpha}}^{(r)}\rb$ is a singleton set whereas $\cl{G}_{{\alpha}}\lb \bs{\gamma}^{(r)},{ \bs{\alpha}}^{(r)}\rb$ need not be. Also, for a given ${ \bs{\alpha}}^{(r)}$, the mapping $\cl{G}( \bs{\gamma}^{(r)},{ \bs{\alpha}}^{(r)})$ is continuous in $\bs{\gamma}^{(r)}$.  Please refer to \cite{joseph2022state} for more details. Using the above properties of the algorithm, we derive the convergence results.

We start with the optimization problem~\eqref{eq:overall_opt} that the algorithm solves. To compute $L(\bs{\theta})=\log p(\bs{Y};\bs{\theta})$, we note that $\bs{Y}$ is Gaussian distributed with zero mean. Given $\bs{\theta}$, we derive $\bs{y}_k={ \bs{\alpha}}_k \bs{A}\sum_{j=1}^{k}\bs{D}^{k-j}\bs{u}_j+\bs{w}_k$ from \eqref{eq:lds}. 
Subsequently, the covariance matrix $\bs{R}_{Y}(\bs{\theta}) \in\bb{R}^{Km\times Km}$ of the vectorized version of $\bs{Y}$ is
\begin{multline}\label{eq:Ry}
\bs{R}_{Y}(\bs{\theta}) =\lb \diag\{{ \bs{\alpha}}\}\otimes \bs{A}\rb\tilde{\bs{D}}\\
\times \lb\bs{I}\otimes \diag\{\bs{\gamma}\}\rb\tilde{\bs{D}}^\T\lb \diag\{{ \bs{\alpha}}\}\otimes \bs{A}^\T\rb,
\end{multline}
where $\diag\{\cdot\}$ represents a diagonal matrix with entries of the argument vector along the diagonal
and 
$\tilde{\bs{D}} \in\bb{R}^{KN\times KN}$ as
\begin{equation}\label{eq:Rx}
\tilde{\bs{D}}  =  \begin{bmatrix}
\bs{I} &\bs{0} & \ldots&\bs{0}\\
\bs{D}&\bs{I} & \ldots&\bs{0}\\
\vdots &&\ddots\\
\bs{D}^{K-1}&\bs{D}^{K-2}&\ldots&\bs{I}
\end{bmatrix} .
\end{equation}
By simplifying \eqref{eq:overall_opt} using \eqref{eq:Ry}, the objective function of an optimization problem equivalent to \eqref{eq:overall_opt} reduces to 
\begin{multline}\label{eq:cost}
L(\bs{\theta}) = -\frac{Km}{2}\log(2\pi)-\frac{1}{2}\log \lv\bs{R}_{Y}(\bs{\theta})+\sigma^2\bs{I}\rv \\
- \frac{1}{2}\bs{y} ^\T\lb \bs{R}_{Y}(\bs{\theta})+\sigma^2\bs{I}\rb^{-1}\bs{y}.
\end{multline}

With the objective function defined, we are now ready to state the convergence result. 
\begin{theorem}\label{thm:convergence_sbl}
Suppose that the noise variance $\sigma^2>0$. Let the iterates generated by the Bayesian state estimation algorithm in ~\cite{joseph2022state}  be $\bs{\theta}^{(r)}\in\bb{R}^{N}\times\bb{H}$.  Then, the sequence of iterates $\{\bs{\theta}^{(r)}\}_{r=1}^{\infty}$ converges to a subset of $\cl{S}_*$ over which $L(\bs{\theta})$ is a constant. Here,
\begin{equation}
\cl{S}_* = \{ \bs{\theta} = \begin{bmatrix}
\bs{\gamma}\in\bb{R}^{N} & { \bs{\alpha}}\in\bb{H}
\end{bmatrix}: \nabla_{\bs{\gamma}} L(\bs{\theta})=\bs{0} \},
\end{equation}
where $\nabla_{\bs{\gamma}}$ denote the gradient with respect to $\bs{\gamma}$.
\end{theorem}
\begin{proof}
See \Cref{app:convergence_sbl}.
\end{proof}

Hence, our generalized result in \Cref{thm:general_convergence} guarantees that the SBL variant in~\cite{joseph2022state} with discrete parameters converges to a stationary point of the ML cost function. 

\section{Conclusion}
We derived the conditions for the convergence of the EM algorithm with discrete unknown parameters. As an illustration, we demonstrated the convergence of the EM-based SBL algorithm outlined in \cite{joseph2022state}, proving its convergence to the set of stationary points of the maximum likelihood cost. Extending the results to the generalized class of Majorization-Minimization algorithms is an interesting future work.

\appendices
\crefalias{section}{appendix}
\section{Proof of \Cref{thm:general_convergence}}\label{app:general_convergence}
Our proof is adapted from the proofs of EM algorithm convergence in \cite{wu1983convergence} and Zangwill's convergence theorem~\cite{zangwill1969nonlinear}. Our proof relies on some properties of the algorithm iterates as listed by the following preliminary lemmas:

\begin{lemma}[{~\cite[Theorem 8.1]{little2019statistical}}]\label{lem:splitL}
The EM algorithm formulation guarantees the following from \eqref{eq:overall_opt}:
\begin{equation}\label{eq:L_split}
L( \bs{\theta})
= Q( \bs{\theta};\bs{\theta}^{(r)}) +
g( \bs{\theta};\bs{\theta}^{(r)}),
\end{equation}
where $Q$ is defined in \eqref{eq:Q_defn} and we define 
\begin{equation}\label{eq:g_defn}
g( \bs{\theta};\bs{\theta}^{(r)}) =  \bb{E}_{\bs{X}  |\bs{Y},\bs{\theta}^{(r)}}\{-\log p(\bs{X}|\bs{Y},\bs{\theta})\}.
\end{equation} 
\end{lemma}
\begin{lemma}[Gibbs' inequality~\cite{kullback1951information}]\label{lem:Gibbs}
For any $\bs{\theta}\in\bs{\Theta}$, the function $g$ defined in \eqref{eq:g_defn} satisfies
$ g(\bs{\theta};\bs{\theta}^{(r)}) \geq g(\bs{\theta}^{(r)};\bs{\theta}^{(r)})$.
\end{lemma}
\begin{lemma}\label{lem:fnt_convergence}
The sequence of function values $\{ L( \bs{\theta}^{(r)}),\; r\geq 1\}$ defined in \eqref{eq:overall_opt} is monotonically non-decreasing and convergent. 
\end{lemma}
\begin{proof}
From \eqref{eq:MStep}, in every iteration $r$,
\begin{equation}\label{eq:Q_increase}
Q(\bs{\theta}^{(r+1)};\bs{\theta}^{(r)}) \geq Q(\bs{\theta}^{(r)};\bs{\theta}^{(r)}).
\end{equation}
From \Cref{lem:splitL,lem:Gibbs},  and \eqref{eq:Q_increase}, $L( \bs{\theta}^{(r+1}) - L(\bs{\theta}^{(r)}) \geq 0$.
Consequently, the sequence $\{ L(\bs{\theta}^{(r)}), r\geq 1\}$ is monotonically non-decreasing, and is bounded from above by Assumption~\ref{assmp:lowerbnd}. Thus, it convergences to a single point. 
\end{proof}
\begin{lemma}\label{lem:incresesL}
If  $\bs{\theta}^{(r)}\notin\cl{S}^*$ for some $r>0$, then we have the relation 
$
L(\bs{\theta}^{(r+1};\bs{\theta}^{(r)})> L(\bs{\theta}^{(r)};\bs{\theta}^{(r)})$.
\end{lemma}
\begin{proof}
Using \Cref{lem:splitL}, we get
\begin{equation}
 \nabla_{\bs{\gamma}}Q( \bs{\theta}^{(r)};\bs{\theta}^{(r)}) + \nabla_{\bs{\gamma}}g( \bs{\theta}^{(r)};\bs{\theta}^{(r)})\neq \bs{0}. \label{eq:nonzero_grad}
\end{equation}
Also, from \Cref{lem:Gibbs}, we have
\begin{equation}
 g( [\bs{\gamma}^{\T} \;\;\;{ \bs{\alpha}}^{(r)\T}]^{\T};\bs{\theta}^{(r)}) \geq g( [\bs{\gamma}^{(r)\T}\;\;\; { \bs{\alpha}}^{(r)\T}]^{^{\T}};\bs{\theta}^{(r)}).
\end{equation}
Hence, $ \nabla_{\bs{\gamma}}g( \bs{\theta}^{(r)};\bs{\theta}^{(r)})=\bs{0}$, and as a result, from \eqref{eq:nonzero_grad}, we have
$\nabla_{\bs{\gamma}}Q( \bs{\theta}^{(r)};\bs{\theta}^{(r)})\neq \bs{0}$. So, $\bs{\theta}^{(r)}$ is not a local $Q( \bs{\theta}^{(r)};\bs{\theta}^{(r)})$. Therefore, from the definition of the M-step update in \eqref{eq:MStep}, we conclude that 
\begin{equation}\label{eq:Q_strict}
Q(\bs{\theta}^{(r+1)};\bs{\theta}^{(r)})> Q(\bs{\theta}^{(r)};\bs{\theta}^{(r)}).
\end{equation} 
Now, we arrive at the desired result by \Cref{lem:splitL,lem:Gibbs}.

\end{proof}
\subsection*{Proof of \Cref{thm:general_convergence}}
Let $\bs{\theta}_*=[
\bs{\gamma}_*^{\T}\;\;\;  { \bs{\alpha}}_*^{\T}\in\bb{H}]^{\T}\in\bs{\Theta}$ be a limit point of the sequence $\{ \bs{\theta}^{(r)}, r\geq 1\}$. From \Cref{lem:fnt_convergence} and Assumption~\ref{assmp:compact}, we know that there exists a subsequence $\{ \bs{\theta}^{(r_j)},j\geq 1\}$ of $\{ \bs{\theta}^{(r)},r\geq 1\}$ such that 
\begin{equation}\label{eq:converge_sub}
\lim_{j\to\infty}\bs{\theta}^{(r_j)} = \bs{\theta}_*=[
\bs{\gamma}_*^{\T}\;\;\;  { \bs{\alpha}}_*^{\T}]^{\T}.
\end{equation}
We next construct another subsequence $\{ \bs{\theta}^{(r_j+1)},\;j\geq 1\}$, which also belongs a compact set due to Assumption~\ref{assmp:compact}. Hence, the new sequence contains a convergent subsequence $\{ \bs{\theta}^{(r_{j_l}+1)},\;l\geq 1\}$ for which there exists $\hat{\bs{\theta}}$ such that
\begin{equation}\label{eq:subseq_limit_2}
\lim_{l\to\infty}  \bs{\theta}^{(r_{j_l}+1)} = \hat{\bs{\theta}}.
\end{equation}

From \eqref{eq:MStep}, we get 
\begin{equation}
\bs{\theta}^{(r_{j_l}+1)} \in \cl{G}(
\bs{\gamma}^{(r_{j_l})}, { \bs{\alpha}}^{(r_{j_l})}).
\end{equation}
Here, by construction, $\{{ \bs{\alpha}}^{(r_{j_l})},l\geq 1\}$ is a subsequence of the convergent sequence $\{ \bs{\alpha}^{(r)},r\geq 1\}$. So it  converges to ${ \bs{\alpha}}_*$, and since the subsequence belongs to a finite set, there exists $L>0$ such that 
\begin{equation}
\bs{\theta}^{(r_{j_l}+1)} \in \cl{G}(
\bs{\gamma}^{(r_{j_l})}, { \bs{\alpha}}_*),\;\; \forall\;l>L.
\end{equation}
Therefore, by Assumption~\ref{assmp:closed}  and \eqref{eq:subseq_limit_2}, we arrive at
\begin{equation}\label{eq:contra_inter}
\hat{\bs{\theta}}  \in \cl{G}\lb 
\lim_{l\to\infty} \bs{\gamma}^{(r_{j_l})}, { \bs{\alpha}}_*
 \rb  = \cl{G}(\bs{\gamma}_*, { \bs{\alpha}}_*),
\end{equation}
due to \eqref{eq:converge_sub}, where $\cl{G}(\bs{\gamma}_*, { \bs{\alpha}}_*)$ is the next iterate of the EM algorithm if the current iterate is $\bs{\theta}_*$. The relation \eqref{eq:contra_inter} implies 
\begin{equation}\label{eq:contra_final}
L(\bs{\theta}_*)  = L(\hat{\bs{\theta}})= L\lb \cl{G}(\bs{\gamma}_*, { \bs{\alpha}}_*)\rb,
\end{equation}
due to the convergence of the sequence $\{ L(\bs{\theta}^{(r)}),r\geq1\}$  as ensured by \Cref{lem:fnt_convergence}.

Furthermore, by \Cref{lem:incresesL}, if $\bs{\theta}_*\notin\cl{S}^*$, then,  $L(\bs{\theta}_*) > L\lb \cl{G}( \bs{\gamma}_*,\bs{\alpha}_*)\rb$. Hence, \eqref{eq:contra_final} holds only if $\bs{\theta}_*\in\cl{S}^*$. Finally, by \Cref{lem:fnt_convergence}, $L(\bs{\theta})$ is a constant over the subset of  $\cl{S}^*$ to which the iterates converge, and the proof is completed.

\section{Proof of \Cref{thm:convergence_sbl}}\label{app:convergence_sbl}

The proof verifies the four assumptions of \Cref{thm:convergence_sbl} hold for the algorithm. We need the following supporting lemmas.

\begin{lemma}[{\cite[Theorem 2.11]{arnuautu2003optimal}}]\label{lem:bound}
Let $\cl{S}$ be an unbounded subset of $\bb{R}^n$, and $f:\cl{S}\to\bb{R}$ is a continuous function, then $f$ is said to be coercive (i.e., $\underset{\lV\bs{s}\rV\to\infty}{\lim}f(\bs{s}) = \infty$) if and only if all of its level sets are compact. 
\end{lemma}

\begin{lemma}[{ \cite[Theorem 5.19, 5.20]{schott2016matrix}}]\label{lem:continous}
The determinant and inverse of a matrix are continuous in its elements. 
\end{lemma}

\begin{lemma}\label{lem:compact}
The function $L( [\bs{\gamma} ^{\T}\;\;\; { \bs{\alpha}^{\T}}]^{\T})$  is continuous and coercive with respect to $\bs{\gamma}$.
\end{lemma}
\begin{proof}
By \Cref{lem:continous}, the determinant and inverse of a matrix are continuous in its elements. Therefore, from \eqref{eq:Ry} and \eqref{eq:cost}, $L( [\bs{\gamma} ^{\T}\;\;\; { \bs{\alpha}^{\T}}]^{\T})$ is a continuous function of $\bs{\gamma}$. Further, $\lV\bs{\gamma}\rV$
goes to $\infty$ if and only if at least one entry of $\bs{\gamma}$ goes to $\infty$. As a result, from \eqref{eq:Ry}, we get 
\begin{equation}
\lim_{\lV\bs{\gamma}\rV\to\infty}\bs{y} ^\T\lb \bs{R}_{Y}\bs{\theta}+\sigma^2\bs{I}\rb^{-1}\bs{y} = \bs{0}.
\end{equation}
This relation leads to the following,
\begin{equation}
\lim_{\lV\bs{\gamma}\rV\to\infty} L( [\bs{\gamma} ^{\T}\;\;\; { \bs{\alpha}^{\T}}]^{\T})  = \lim_{\lV\bs{\gamma}\rV\to\infty} \log \lv\bs{R}_{Y}\bs{\theta}+\sigma^2\bs{I}\rv =\infty.
\end{equation} 
Hence, we conclude that $L( [\bs{\gamma} ^{\T}\;\;\; { \bs{\alpha}^{\T}}]^{\T})$  is a coercive function of $\bs{\gamma}$, and the proof is complete.
\end{proof}

\begin{lemma}[{\cite[Theorem 4.2.1]{strichartz2000way}}]\label{lem:max_contin}
Let $(\cl{X},\tau)$ be a topological space and the functions $q_1,q_2,q_3,\ldots,q_p :\cl{X}\to\bb{R}$ be continuous, for some finite $p>0$. Then, the function $q_{\max}$ defined as
\begin{equation}
q_{\max}=\max\{q_1,q_2,\ldots,q_p\}    
\end{equation}
 is continuous.
\end{lemma}

\subsection*{Proof of \Cref{thm:convergence_sbl}}
We start with verifying the first assumption of \Cref{thm:general_convergence}. To this end, we notice that covariance matrix $\bs{R}_{Y}(\bs{\theta})$ is positive semidefinite, and thus, we have
\begin{align}
\log \lv\bs{R}_{Y}(\bs{\theta})+\sigma^2\bs{I}\rv&\geq \log \lv\sigma^2\bs{I}\rv\\
\bs{y} ^\T\lb \bs{R}_{Y}(\bs{\theta})+\sigma^2\bs{I}\rb^{-1}\bs{y}&> 0,
\end{align}
for any $\bs{y}\in\bb{R}^{Km}$. Therefore, \eqref{eq:cost} leads to the following:
\begin{equation}
L(\bs{\theta}) \leq -\frac{Km}{2} \log\lb 2\pi\sigma^2\rb,
\end{equation}
and Assumption~\ref{assmp:lowerbnd} holds. 

To verify the second assumption of \Cref{thm:convergence_sbl}, we note that $\bb{H}$ is a compact set, and it is sufficient to show that $\bs{\gamma}^{(r)}$ belongs to a compact set. For this, we define
\begin{multline}
\cl{S} \triangleq \Big\{ \bs{\gamma}\in\bb{R}^n:\; \exists \;{ \bs{\alpha}}\in\bb{H}\; \text{such that}\\
L( [\bs{\gamma} ^{\T}\;\;\; { \bs{\alpha}^{\T}}]^{\T})\geq L ( \bs{\theta}^{(0)})\Big\}.
\end{multline}
Also, by \Cref{lem:fnt_convergence}, $L(\bs{\theta}^{(r)})\leq L (\bs{\theta}^{(0)})$, for any integer $r>0$.
Therefore, $\bs{\gamma}^{(r)}\in \cl{S}$. Consequently,  it is enough to show that $\cl{S}$ is a compact set. To this end, we rewrite $\cl{S}$ as a finite union of level sets of $L$,
\begin{equation}
\cl{S} = \bigcup_{{ \bs{\alpha}}\in\bb{H}}\Big\{ \bs{\gamma}\in\bb{R}^n:\; L( [\bs{\gamma} ^{\T}\;\;\; { \bs{\alpha}^{\T}}]^{\T})\geq L ( \bs{\theta}^{(0)})\Big\}.
\end{equation}
Since a finite union of compact sets is compact, we need to show that the level sets of $L( [\bs{\gamma} ^{\T}\;\;\; { \bs{\alpha}^{\T}}]^{\T})$ for a fixed value of ${ \bs{\alpha}}$ are compact.  Invoking \Cref{lem:bound}, these level sets are compact if $L( [\bs{\gamma} ^{\T}\;\;\; { \bs{\alpha}^{\T}}]^{\T})$ is continuous and coercive with respect to $\bs{\gamma}$. Then, by \Cref{lem:compact}, Assumption~\ref{assmp:compact} holds.

To check the third assumption of \Cref{thm:convergence_sbl}, we verify if the following holds:
\begin{equation}
\lim_{r\to\infty} \cl{G}(
\bs{\gamma}^{(r)} , { \bs{\alpha}}_*) = \cl{G}\lb
\lim_{r\to\infty} \bs{\gamma}^{(r)} , { \bs{\alpha}}_*
\rb,
\end{equation}
when the limits exist. For this, we recall from \eqref{eq:mapping_sbl} that  $\cl{G}_{\gamma}( \bs{\gamma}^{(r)},{ \bs{\alpha}}_*)$ is a singleton set and continuous in $\bs{\gamma}^{(r)}$. So, 
\begin{equation}
\lim_{r\to\infty }\cl{G}_{\gamma}( \bs{\gamma}^{(r)},{ \bs{\alpha}}_*) = \cl{G}_{\gamma}\lb \lim_{r\to\infty } \bs{\gamma}^{(r)},{ \bs{\alpha}}_*\rb.
\end{equation}
We next complete the proof by establishing that the limit $\lim_{r\to\infty }\cl{G}_{{\alpha}}\lb \bs{\gamma}^{(r)}, { \bs{\alpha}}_*\rb=\cl{G}_{{\alpha}}\lb \lim_{r\to\infty }\bs{\gamma}^{(r)},{ \bs{\alpha}}_*\rb$. For this, we consider the sequence $\{ Q_{\max}\lb\bs{\gamma}^{(r)}\rb  \}_{r=1}^{\infty}$ where
\begin{equation}
Q_{\max}(\bs{\gamma}^{(r)}) = \underset{{ \bs{\alpha}}\in\{0,1\}^K}{\max}\; q(\bs{\gamma}^{(r)},\bs{\alpha}),
\end{equation}
where $q(\bs{\gamma}^{(r)},\bs{\alpha})=Q([ \bs{\gamma}^{(r)\T}\;\;\;{ \bs{\alpha}}^{\T}]^{\T};[\bs{\gamma}^{(r)\T}\;\;\;{ \bs{\alpha}}_*^{\T}]^{\T})$
from \eqref{eq:MStep}. We notice that $\{0,1\}^K$ is a finite set. Also, since $p(\bs{Y},\bs{X}|\bs{\theta})$ and $p(\bs{Y},\bs{X}|\bs{\theta}^{(r)})$ are Gaussian, $q(\bs{\gamma}^{(r)},\bs{\alpha})$ is a continuous function of $\bs{\gamma}^{(r)}$. Then, invoking \Cref{lem:max_contin}, we obtain that $Q_{\max}(\bs{\gamma}^{(r)})$ is a continuous function of $\bs{\gamma}^{(r)}$. Therefore, with $\bbm{1}$ being the indicator function, we derive that for any ${ \bs{\alpha}}$,
\begin{align}
\lim_{r\to\infty} \bbm{1}\lc{ \bs{\alpha}}\in\cl{G}_{{\alpha}}( \bs{\gamma}^{(r)},{ \bs{\alpha}}_*)\rc \notag\\
&\hspace{-3.5cm}= \lim_{r\to\infty} \bbm{1}\lc  q(\bs{\gamma}^{(r)},\bs{\alpha}) =  Q_{\max}(\bs{\gamma}^{(r)})\rc\\
&\hspace{-3.5cm}= \bbm{1}\lc  q\lb\lim_{r\to\infty}\bs{\gamma}^{(r)},\bs{\alpha}\rb =  Q_{\max}\lb\lim_{r\to\infty}\bs{\gamma}^{(r)}\rb\rc\label{eq:interlemma}\\
&\hspace{-3.5cm}= \bbm{1}\lc  { \bs{\alpha}}\in\cl{G}_{{\alpha}}\lb \lim_{r\to\infty}\bs{\gamma}^{(r)}, { \bs{\alpha}}_*\rb \rc,
\end{align}
where \eqref{eq:interlemma} uses the continuity of $q$ and $Q_{\max}$.
As a result, 
\begin{align}
\lim_{r\to\infty} \cl{G}_{{\alpha}}(
\bs{\gamma}^{(r)} , { \bs{\alpha}}_*
 )\notag\\
 &\hspace{-2.6cm}=  \lc{ \bs{\alpha}}\in\{0,1\}^K:  \bbm{1}\lc {\bs{\alpha}}\in\cl{G}_{{\alpha}}\lb \lim_{r\to\infty}\bs{\gamma}^{(r)},{ \bs{\alpha}}_*\rb\rc = 1 \rc\\
 &\hspace{-2.6cm} =  \cl{G}_{{\alpha}}\lb
\lim_{r\to\infty}\bs{\gamma}^{(r)} , { \bs{\alpha}}_*
 \rb.
\end{align}
Hence,  Assumption~\ref{assmp:closed} holds.

Finally, we verify the fourth assumption of \Cref{thm:general_convergence}. The dependence of $\bs{\gamma}^{(r)}$ and $\bs{\gamma}$ on $Q$ is via distributions  $p(\bs{Y},\bs{X}|\bs{\theta})$ and $p(\bs{Y},\bs{X}|\bs{\theta}^{(r)})$, which are Gaussian. The function is computed using Kalman smoothing that involves only continuous functions of $\diag\{\bs{\gamma}\}$ and $\diag\{\bs{\gamma}^{(r)}\}$. Thus,  Assumption~\ref{assmp:contin} holds, and the proof is complete.

\bibliographystyle{IEEEtran}
\bibliography{Missing_cite}
\end{document}